\documentclass[3p,12pt,sort&compress]{elsarticle}
\pdfoutput=1
\usepackage[utf8]{inputenc}

\usepackage{amsthm}
\usepackage{mathtools}
\usepackage{lmodern}
\usepackage[T2A,T1]{fontenc}
\usepackage{amsfonts}
\usepackage{amssymb}
\usepackage{bbm}
\DeclareMathAlphabet{\mathbfi}{OML}{cmm}{b}{it}

\let\originalleft\left
\let\originalright\right
\renewcommand{\left}{\mathopen{}\mathclose\bgroup\originalleft}
\renewcommand{\right}{\aftergroup\egroup\originalright}

\makeatletter
\newenvironment{equations}[1][]{\subequations\ifx\relax#1\relax\else\label{#1}\fi\align\ignorespaces}{\endalign\ignorespacesafterend\endsubequations}
\def\@spliteq#1{\begin{equation}\begin{split}#1\end{split}\end{equation}}
\def\@spliteqstar#1{\begin{equation*}\begin{split}#1\end{split}\end{equation*}}
\def\splitequation{\collect@body\@spliteq}
\expandafter\def\csname splitequation*\endcsname{\collect@body\@spliteqstar}

\expandafter\def\csname endsplitequation*\endcsname{\ignorespacesafterend}
\makeatother

\allowdisplaybreaks

\usepackage{xcolor}
\usepackage[colorlinks=true,linktocpage=true,pdfa=true,unicode=true,bookmarksopen=true]{hyperref}

\renewcommand{\vec}[1]{{\ifnum9<1#1\mathbf{#1}\else\ifcat\noexpand#1\relax\boldsymbol{#1}\else\mathbfi{#1}\fi\fi}}
\newcommand{\mathe}{\mathrm{e}}
\newcommand{\mathi}{\mathrm{i}}
\newcommand{\total}{\mathop{}\!\mathrm{d}}
\newcommand{\eqend}[1]{\,#1}
\newcommand{\defeq}{\coloneq}
\newcommand{\abs}[1]{{\left\lvert{#1}\right\rvert}}
\newcommand{\norm}[1]{{\left\lVert{#1}\right\rVert}}
\newcommand{\sgn}{\operatorname{sgn}}
\newcommand{\1}{\mathbbm{1}}

\newtheorem{theorem}{Theorem}[section]
\newtheorem{lemma}[theorem]{Lemma}
\newtheorem{proposition}[theorem]{Proposition}

\newtheorem{definition}[theorem]{Definition}

\allowdisplaybreaks

\journal{Journal of Geometry and Physics}

\begin{document}

\title{A proposal for the algebra of a novel noncommutative spacetime}

\author[1,2]{Markus B. Fr{\"o}b}
\ead{markus.froeb@fau.de}
\author[1]{Albert Much}
\ead{much@itp.uni-leipzig.de}
\author[3]{Kyriakos Papadopoulos\corref{cor1}}
\ead{kyriakos.papadopoulos@ku.edu.kw}

\affiliation[1]{organization={Institut f{\"u}r Theoretische Physik, Universit{\"a}t Leipzig}, addressline={Br{\"u}derstra{\ss}e 16}, postcode={04103}, city={Leipzig}, country={Germany}}
\affiliation[2]{organization={Department Mathematik, Friedrich--Alexander--Universit{\"a}t Erlangen--Nuremberg}, addressline={Cauerstraße~11}, postcode={91058}, city={Erlangen}, country={Germany}}
\affiliation[3]{organization={Department of Mathematics, College of Science, Kuwait University}, addressline={Sabah Al Salem University City, P.O.~Box 5969}, postcode={Safat 13060}, city={Shadadiya}, country={Kuwait}}

\cortext[cor1]{Corresponding author}

\begin{abstract}
We investigate the quantum structure of spacetime at fundamental scales via a novel, Lorentz-invariant noncommutative coordinate framework. Building on insights from noncommutative geometry, spectral theory, and algebraic quantum field theory, we systematically construct a quantum spacetime algebra whose geometric and causal properties are derived from first principles. Using the Weyl algebra formalism and the Gelfand--Naimark--Segal (GNS) construction, we rigorously define operator-valued coordinates that respect Lorentz symmetry and encode quantum gravitational effects through nontrivial commutation relations. We show how the emergent quantum spacetime exhibits minimal length effects, which deliver both classical Minkowski distances and quantum corrections proportional to the Planck length squared. Furthermore, we establish that noncommutativity respects a fuzzy form of causality, where the quantum causal structure gives back the light cone in the classical limit, vanishing for spacelike separations and encoding a time orientation for timelike intervals.
\end{abstract}


\maketitle

\clearpage
\tableofcontents

\section{Introduction}
\label{sec:intro}

The fundamental structure of spacetime at quantum scales remains a profound open question in theoretical physics. Recent advances in noncommutative geometry have revealed deep connections between spectral geometry~\cite{connes1994} and the emergence of Lorentzian structure from more fundamental algebraic frameworks, see for example Refs.~\cite{LNCG1, LNCG2, LNCG3} and references therein. In this work, we investigate a novel approach to quantizing coordinates that preserves essential geometric and causal properties while incorporating quantum gravitational effects. For a historic overview of the subject with many references, we refer the reader to the work of Maresca~\cite{maresca2025}, and for a collection of recent works to Refs.~\cite{NC1,NC2,NC3,NC4,NC5,NC6,NC7,NC8}.

An important approach to noncommutative geometry is the spectral triple framework of Connes, which provides a unified framework where geometric information is encoded in purely algebraic and analytical structures~\cite{connes1994}. This framework has found remarkable applications in particle physics~\cite{connesmarcolli2008}, where the noncommutative standard model emerges naturally from an almost-commutative spectral triple that combines continuous spacetime with a discrete internal structure. A spectral triple $(\mathcal{A}, \mathcal{H}, \mathcal{D})$ consists of an involutive algebra $\mathcal{A}$ acting on a Hilbert space $\mathcal{H}$, together with a self-adjoint operator $\mathcal{D}$ satisfying some technical assumptions. The simplest examples of spectral triples are the ones canonically associated to smooth compact Riemannian spin manifolds $\mathcal{M}$, where $\mathcal{A}$ is the algebra of Lipschitz maps $f \colon \mathcal{M} \to \mathbb{C}$ with the supremum norm $\norm{ \cdot }_\infty$, $\mathcal{H} = L^2(\mathcal{M},S)$ is the space of square integrable sections of the spinor bundle $S$ over $\mathcal{M}$, and $\mathcal{D}$ is the Dirac operator associated with the Levi--Civita connection on $\mathcal{M}$. $\mathcal{A}$ acts on $\mathcal{H}$ by pointwise multiplication, and the pure states on $\mathcal{A}$ are the evaluation functionals $p\colon f \mapsto f(p)$. It is then easy to show (see e.g., Refs.~\cite{connes1994,ConnesDistance,graciabondiavarillyfigueroa2001,franco2010}) that the Riemannian distance between $p$ and $q$ can be obtained as the supremum
\begin{equation}
\label{eq:spectral_triple_distance}
d(p,q) = \sup_{f \in \mathcal{A}\colon \norm{ [ \mathcal{D}, f ] \leq 1 }} \abs{ p(f) - q(f) } \eqend{.}
\end{equation}
This formula extends to all states on $\mathcal{A}$, and one may employ it to define a distance also in the case where $\mathcal{A}$ is noncommutative.\footnote{In fact, this distance is nothing else but the Wasserstein distance of order 1 between $p$ and $q$, in the dual formulation of Kantorovich~\cite{rieffel1999,martinetti2018}. Alternatively, one can also write $d(p,q)$ as an infimum~\cite{dandreamartinetti2021}.}

However, the extension to Lorentzian geometry presents fundamental challenges. In the Riemannian case, it has been proven (through the Gel'fand--Naimark correspondence~\cite{GelfandNaimark}) that $\mathcal{M}$ is homeomorphic to the space of (unitary equivalence classes of) irreducible representations of the completion of the algebra $\mathcal{A}$ of smooth functions on $\mathcal{M}$, equipped with the topology of pointwise convergence (the weak $*$ topology), and the points of $\mathcal{M}$ arise as pure states on $\mathcal{A}$. On the other hand, the Lorentzian case is significantly more subtle due to the indefinite signature of the metric. The pioneering work of Franco, Eckstein, and others established a spectral formulation of the Lorentzian distance formula~\cite{LNCG1,LNCG2,LNCG3,LNCG4,LNCG5,LNCG6,LNCG7,LNCG8,LNCG9}, which relies on Krein space structures. A related issue is the relation between the manifold topology on a Lorentzian spacetime and the Lorentzian distance and causal relations between points, which again is subtle due to the indefinite signature of the metric; see Minguzzi's review~\cite{MinguzziTopology} for a comprehensive analysis and conditions under which one can resolve this relation.

Parallel to these geometric developments, the quest for Lorentz-invariant noncommutative coordinates has been actively pursued. A widely studied model of noncommutative coordinates is the Doplicher--Fredenhagen--Roberts (DFR) algebra~\cite{DFR1995}, where the coordinate operators $\hat{x}^\mu$ have the nontrivial commutator
\begin{equation}
\label{eq:dfr}
[ \hat{x}^\mu, \hat{x}^\nu ] = \mathi \Theta^{\mu\nu}
\end{equation}
with a constant skew-symmetric matrix $\Theta^{\mu\nu}$, whose entries are of the order of the (squared) Planck length. However, a constant matrix is compatible with Lorentz covariance only in two spacetime dimensions or if one considers a twisted Poincar\'e symmetry~\cite{chaichiankulishnishijimatureanu2004}. Constructing explicit realizations that maintain manifest Lorentz invariance while avoiding pathological behaviors remains challenging.

Another widely studied model is the Snyder spacetime~\cite{snyder1947}, where the commutator of coordinate operators is given by
\begin{equation}
\label{eq:snyder}
[ \hat{x}^\mu, \hat{x}^\nu ] = \mathi a^2 \hat{M}^{\mu\nu}
\end{equation}
with a parameter $a > 0$ (again of the order of the squared Planck length) and the generators of Lorentz transformations $\hat{M}^{\mu\nu}$. As a deformation of the Poincaré algebra, this is explicitly Lorentz invariant. However, explicit realizations of this algebra are inextrinsicably linked to the phase space of particles transforming in its representations. Therefore, it seems to us that this approach does not yield a quantity that is intrinsic to spacetime itself, but is more akin to the nonvanishing commutation relations obtained for position operators of relativistic particles~\cite{dekerfbaeuerle1971}. On the other hand, it is highly suitable for studies of quantum field theories on such spacetimes and can be extended to other curved spacetimes such as de~Sitter space; see \cite{kosinskilukierskimaslanka2000,kowalskiglikmansmolin2004,arzanomarciano2007,mignemi2010,muchvergara2018,mercatisergola2018b,franchinovinasmignemi2020} for a small sample of a vast amount of literature.

Building upon our previous investigations of Lorentz-invariant noncommutative coordinates~\cite{froebmuchpapadopoulos2023a,froebmuchpapadopoulos2023b,froebmuchpapadopoulos2023c}, this paper addresses a fundamental question: given a set of coordinates with specified Lorentz-invariant commutation relations, how should we properly quantize this system while respecting both the algebraic structure and underlying geometric properties? Put differently, how can we construct a quantum spacetime out of the algebra? And moreover, what does a quantum spacetime mean? Drawing inspiration from the Ehlers--Pirani--Schild framework~\cite{ehlerspiranischild1972}, one expects that a suitable algebraic construction should yield a distance functional and a causal structure, ensuring that both geometric and causal properties are intrinsically incorporated. Moreover, the distance functional should yield the classical (Lorentzian) spacetime distance as its leading contribution, with additional terms naturally interpretable as quantum corrections.

The central challenge lies in reconciling several competing requirements. First, while Riemannian manifolds can be reconstructed from their algebra of functions, the Lorentzian case presents fundamental topological obstructions~\cite{MinguzziTopology}. As noted by Minguzzi and others, standard manifold topology ``ignores the Lorentzian structure'' and fails to capture essential causal relationships. The quantization process must thus preserve the indefinite Lorentzian signature, while at the same time construct a Hilbert space with positive-definite metric. Our approach centers on the recognition that proper quantization of Lorentzian noncommutative coordinates requires auxiliary structures to define the underlying Hilbert space, but crucially, these can be constructed out of the indefinite inner product that incorporates the Lorentzian geometry. This resolves the apparent tension between the need for positive-definite structures in quantum mechanics and the indefinite nature of Lorentzian geometry. As our construction shows, a nonvanishing Lorentz-invariant commutator between coordinates can be realized without reference to particles or fields that propagate on the spacetime. In this sense, our proposal thus gives an intrinsic property of spacetime itself, more akin to the DFR model than the Synder one (or its generalizations).

The remainder of this work is organized as follows: in section~\ref{sec:quantum_nc} we recall the construction of Lorentz-invariant noncommutative coordinates in our earlier work~\cite{froebmuchpapadopoulos2023a,froebmuchpapadopoulos2023b,froebmuchpapadopoulos2023c}. Paralleling the construction of free bosonic quantum fields, we then define a corresponding Weyl algebra in subsection~\ref{subsec:weyl}, and show that coordinate operators can be obtained as limits in a suitable topology. To obtain a positive state on this algebra and thus a Hilbert space, in subsection~\ref{subsec:krein} we recall the Krein space construction, and in subsection~\ref{subsec:dm_state} we define the analog of the Dereziński--Meissner state. In section~\ref{sec:nc_spacetime}, we then show how a distance functional (in subsection~\ref{subsec:distance}) and the causal structure (in subsection~\ref{subsec:causal}) can be extracted from the quantum theory. In particular, we show how both a minimal distance and fuzzy causality naturally emerge from our proposal, without the need for ad-hoc postulates. We conclude in section~\ref{sec:discussion}.

\section{Quantum theory of noncommutative coordinates}
\label{sec:quantum_nc}

Let us recall how the noncommutative coordinates were derived in our earlier work~\cite{froebmuchpapadopoulos2023a,froebmuchpapadopoulos2023b} in the context of perturbative quantum gravity, seen as an effective field theory of quantum gravity~\cite{burgess2003}. We started by considering Minkowski spacetime with Cartesian coordinates $x^\mu = (t, \vec{x})^\mu$, $x^2 = - t^2 + \vec{x}^2$, and introduced metric perturbations $h_{\mu\nu}$ around the Minkowski background, decomposing the full metric into $g_{\mu\nu} = \eta_{\mu\nu} + \kappa h_{\mu\nu}$ with perturbative parameter
\begin{equation}
\label{eq:kappa_def}
\kappa \defeq \sqrt{ \frac{16 \pi \hbar G_\text{N}}{c^4} } \sim \ell_\mathrm{Pl}
\end{equation}
proportional to the Planck length $\ell_\mathrm{Pl}$. Linearized diffeomorphisms generated by a vector field $\xi^\mu$ act as symmetry transformations on $h_{\mu\nu}$ according to $\delta_\xi h_{\mu\nu} = \partial_\mu \xi_\nu + \partial_\nu \xi_\mu$, and observables invariant under this symmetry can be constructed using the relational approach (see, e.g., Refs.~\cite{giddingsmarolfhartle2006,gieselhofmannthiemannwinkler2010,tambornino2012,gieselthiemann2015,khavkine2015,goellerhoehnkirklin2022,giddings2025,ferrerothiemann2025} and references therein). In practice, one has to construct a dynamical coordinate system $X^{(\mu)}(x)$ transforming as $\delta_\xi X^{(\mu)} = \xi^\rho \partial_\rho X^{(\mu)}$ and which in the classical limit $\hbar \to 0$, hence $\kappa \to 0$ reduces to the background coordinates $x^\mu$, and evaluate the quantity of interest in this system, see for example~\cite{brunettifredenhagenrejzner2016,froeblima2018,baldazzifallsferrero2022,froeblima2023}. In the linearized theory, we obtained the $X^{(\mu)}$ as local causal functionals of the metric perturbation, and under quantization they inherit nontrivial commutators from the quantized metric perturbation $h_{\mu\nu}$. As a result, at leading order we obtain the nonvanishing commutator
\begin{equation}
\label{eq:coord_commutator}
\left[ X^{(\mu)}(x), X^{(\nu)}(x') \right] = - \mathi \kappa^2 \frac{\eta^{\mu\nu}}{8 \pi} \sgn(t-t') \Theta\left[ -(x-x')^2 \right] \1
\end{equation}
for the field-dependent coordinates, where $\Theta\left[ -(x-x')^2 \right]$ is equal to $1$ for timelike separations and vanishes for spacelike ones.

This result shows that noncommutative geometry emerges naturally from the quantum nature of gravity, without the need for \emph{ad hoc} assumptions. The induced noncommutativity is fully Lorentz-invariant and compatible with microcausality, being zero for spacelike separations and constant within the light cone. The scale of noncommutativity is set by the squared Planck length $\ell_\mathrm{Pl}^2$, which arises from the underlying quantum structure of gravity. Higher-order corrections to the commutator~\eqref{eq:coord_commutator} would arise when considering perturbative corrections beyond linear order. These contributions become relevant as one approaches scales closer to the Planck length or considers strong-field regimes, potentially modifying the functional form and magnitude of the noncommutativity. However, within the effective field theory framework and for physical situations far from such extreme regimes, the leading linear order provides a reliable and calculable description. For clarity and technical tractability, we thus restrict the present analysis to linear order, postponing the study of higher-order effects to future work where their resummation and physical significance can be addressed in more detail. An intriguing and deceptively simple question arises at this point, one whose investigation opens up a broad landscape of possibilities: Given a noncommutative algebra of the form~\eqref{eq:coord_commutator}, by what procedure can we reconstruct the corresponding noncommutative spacetime? Moreover, which structural features or data of this space are essential for a meaningful and physically relevant formulation of its geometry and causal properties? That is, if we forget how the commutator~\eqref{eq:coord_commutator} was derived, is it possible to reconstruct the noncommutative Minkowski spacetime from a corresponding noncommutative algebra?

\subsection{The Weyl algebra}
\label{subsec:weyl}

A rigorous approach inspired by algebraic quantum field theory is to treat quantization at the level of the algebra $\mathcal{A}$, and then systematically lift the abstract algebraic structure to a concrete representation of $\mathcal{A}$ on a Hilbert space $\mathcal{H}$. The standard procedure to accomplish this transition is the Gelfand--Naimark--Segal (GNS) construction for $C^*$ algebras, which provides a canonical way to obtain a cyclic $\star$-representation of $\mathcal{A}$ as bounded operators acting on a Hilbert space, given a state $\omega\colon \mathcal{A} \to \mathbb{C}$.\footnote{The GNS construction can also be applied more generally to unital $\star$-algebras $\mathfrak{A}$ equipped with a positive linear functional $\omega\colon \mathfrak{A} \to \mathbb{C}$, which gives a representation in terms of not necessarily bounded operators, see for example \cite[Chapter 4.4]{schmuedgen2020} and~\cite{khavkinemoretti2015}.} Since the commutator~\eqref{eq:coord_commutator} is proportional to the identity operator, any irreducible representation of the $X^{(\mu)}$ themselves results in unbounded operators\footnote{The proof is analogous to the one for the canonical commutation relations.}, and to obtain a $C^*$ algebra one has to consider the corresponding Weyl operators $W(f)$. These are indexed by a collection of real, smooth, compactly supported functions $f_{(\mu)} \in C_0^\infty(\mathbb{R}^4)$, $\mu \in \{0,\ldots,3\}$, which form a real linear space $H$. The pair $(H,\sigma)$ with the symplectic bilinear form
\begin{equation}
\label{eq:sigma_def}
\sigma(f,g) \defeq - \frac{\kappa^2}{8 \pi} \iint \eta^{\mu\nu} f_{(\mu)}(x) g_{(\nu)}(x') \sgn(t-t') \Theta\left[ -(x-x')^2 \right] \total^4 x \total^4 x' = - \sigma(g,f)
\end{equation}
is a symplectic space, from which the Weyl algebra $\mathrm{CCR}(H,\sigma)$ is constructed as follows:
\begin{definition}[Weyl Algebra]
\label{def:weylalgebra}
The Weyl algebra $\mathrm{CCR}(H,\sigma)$ over $(H,\sigma)$ is the $C^*$ algebra generated by elements $W(f)$ with $f \in H$, with the involution given by $[ W(f) ]^* \defeq W(-f)$ and the product by $W(f) W(g) \defeq W(f+g) \exp\left[ - \frac{\mathi}{2} \sigma(f,g) \right]$ (the Weyl relations).
\end{definition}
Since $\sigma(0,g) = 0$, the Weyl algebra is unital with unit given by $\1 \defeq W(0)$, which in particular implies that $\norm{ W(f) }^2 = \norm{ [ W(f) ]^* W(f) } = \norm{ W(-f) W(f) } = \norm{ \1 } = 1$ with the $C^*$ norm $\norm{ \cdot }$.

As it is well known, the topology induced by the $C^*$ norm (the operator norm topology) is a very fine one since we have
\begin{lemma}
For any $f \neq 0$ for which $\sigma(f,g) \neq 0$ for at least one $g \in H$, we have $\norm{ W(f) - \1 } = 2$.
\end{lemma}
\begin{proof}
By the triangle inequality we have $\norm{ W(f) - \1 } \leq \norm{ W(f) } + \norm{ \1 } = 2$, so it only remains to prove the converse inequality. Recall that the spectrum of an element $a \in \mathcal{A}$ in a unital $C^*$ algebra $\mathcal{A}$ is given by $\operatorname{spec} a \defeq \{ \lambda \in \mathbb{C} \colon a - \lambda \1 \ \text{is not invertible} \}$. If $\abs{\lambda} > \norm{ a }$, then the series
\begin{equation}
- \sum_{n=0}^\infty \lambda^{-n-1} a^n
\end{equation}
is convergent in norm and sums to $( a - \lambda \1 )^{-1}$, hence it follows that $\operatorname{spec} a \subseteq \{ \lambda \in \mathbb{C} \colon \abs{\lambda} \leq \norm{a} \}$. If $a$ is invertible, $0 \neq \operatorname{spec} a$ and it follows straightforwardly that $\operatorname{spec} a^{-1} = \{ \lambda \in \mathbb{C} \colon \lambda^{-1} \in \operatorname{spec} a \}$. Since $\norm{ W(f) } = \norm{ W(-f) } = 1$ and $[ W(f) ]^{-1} = W(-f)$, we have $\operatorname{spec} W(f) \subseteq \mathbb{S}^1 \defeq \{ \mathe^{\mathi \lambda} \colon \lambda \in [0,2\pi) \}$, i.e., the Weyl operators are unitary. Using the Weyl relations, we obtain
\begin{equation}
W(g) W(f) [ W(g) ]^* = W(f) \exp\left[ \mathi \sigma(f,g) \right] \eqend{,}
\end{equation}
and since for invertible $b \in \mathcal{A}$ we have $\operatorname{spec} a = \operatorname{spec} ( b a b^{-1} )$, it follows that the spectrum of $W(f)$ is invariant under arbitrary rotations whenever $\sigma(f,g) \neq 0$ for one $g \in H$ (and hence all $\lambda g$ with $\lambda \in (0,\infty)$), such that we have equality $\operatorname{spec} W(f) = \mathbb{S}^1$. It follows that $\operatorname{spec} ( W(f) - \1 ) = \{ \mathe^{\mathi \lambda} - 1 \colon \lambda \in [0,2\pi) \} \subseteq \{ \lambda \in \mathbb{C} \colon \abs{\lambda} \leq \norm{ W(f) - \1 } \}$, and thus $\norm{ W(f) - \1 } \geq 2$.
\end{proof}

It follows analogously that also $\norm{ W(f) - W(g) } = 2$ for all $f \neq g$, such that the operator norm topology on the Weyl algebra is (uniformly) discrete. This topology is thus unsuitable for physical reasons, where one would expect continuity. In particular, if we start with a continuous topology on $H$ (for example the one induced by an $L_p$ or a Sobolev norm), we would like to define a continuous topology also on the Weyl algebra which is equivalent to the initial one on $H$, at least in the classical limit $\hbar \to 0$ where $\sigma = 0$. Such a topology can be defined for the GNS representation of the Weyl algebra, to which we now turn.

\begin{definition}
A linear functional $\phi\colon \mathcal{A} \to \mathbb{C}$ on a $C^*$ algebra $\mathcal{A}$ is called a state if $\phi(\1) = 1$ (normalization) and $\phi\left( a^* a \right) \geq 0$ for all $a \in \mathcal{A}$ (positivity).
\end{definition}
We then have
\begin{theorem}
Given a real linear functional $\mu_1 \colon H \to \mathbb{R}$ and a real symmetric positive bilinear functional $\mu_2 \colon H \times H \to [0,\infty)$, the functional $\phi$ on $\mathrm{CCR}(H,\sigma)$ given by
\begin{equation}
\label{eq:phi_wf}
\phi\left( W(f) \right) = \exp\left[ \mathi \mu_1(f) - \frac{1}{2} \mu_2(f,f) \right]
\end{equation}
defines a state if $\mu_2(f,f) \mu_2(g,g) \geq \frac{1}{4} \sigma(f,g)^2$. 
\end{theorem}
\begin{proof}
Since the normalization condition is fulfilled, we only need to check positivity. Consider thus an element $a = \sum_k \alpha_k W(f_k) \in \mathrm{CCR}(H,\sigma)$ with all $f_i$ linearly independent. By the Weyl relations, we have
\begin{splitequation}
\phi\left( a^* a \right) &= \sum_{k,l} \alpha_k^* \alpha_l \phi\left( W(-f_k) W(f_l) \right) \\
&= \sum_{k,l} \alpha_k^* \alpha_l \exp\left[ \frac{\mathi}{2} \sigma(f_k,f_l) \right] \phi\left( W(-f_k+f_l) \right) \\
&= \sum_{k,l} \alpha_k^* \alpha_l \exp\left[ \frac{\mathi}{2} \sigma(f_k,f_l) + \mathi \mu_1(-f_k+f_l) - \frac{1}{2} \mu_2(-f_k+f_l,-f_k+f_l) \right] \\
&= \sum_{k,l} \left[ \alpha_k \, \mathe^{\mathi \mu_1(f_k) - \frac{1}{2} \mu_2(f_k,f_k)} \right]^* \left[ \alpha_l \, \mathe^{\mathi \mu_1(f_l) - \frac{1}{2} \mu_2(f_l,f_l)} \right] \exp\left[ \mu_2(f_k,f_l) + \frac{\mathi}{2} \sigma(f_k,f_l) \right] \eqend{,}
\end{splitequation}
which should be nonnegative. This is the case if the Hermitean matrix $M$ with entries $M_{kl} = \exp\left[ \mu_2(f_k,f_l) + \frac{\mathi}{2} \sigma(f_k,f_l) \right]$ is positive semidefinite, which means that all its eigenvalues are nonnegative. In turn, this holds when all eigenvalues of the Hermitean matrix $N$ with entries $N_{kl} = \mu_2(f_k,f_l) + \frac{\mathi}{2} \sigma(f_k,f_l)$ are nonnegative, which is the case if $N$ is the Gram matrix of the vectors $\{ f_k \}$ with respect to the scalar product $\left( \cdot, \cdot \right)_1$, defined for real $f$ by
\begin{equation}
\label{eq:hilbert_1_scalarprod}
\left( \cdot, \cdot \right)_1 = \mu_2 + \frac{\mathi}{2} \sigma
\end{equation}
and extended to complex $f$ by sesquilinearity. That is, we need
\begin{splitequation}
\left( f \pm \mathi g, f \pm \mathi g \right)_1 &= \mu_2(f,f) + \mu_2(g,g) \mp \sigma(f,g) \geq 0 \\
&\Leftrightarrow \mu_2(f,f) + \mu_2(g,g) \geq \abs{\sigma(f,g)} \\
&\Leftrightarrow \left[ \sqrt{ \mu(f,f) } - \sqrt{ \mu(g,g) } \right]^2 + 2 \sqrt{ \mu(f,f) \mu(g,g) } \geq \abs{\sigma(f,g)} \eqend{,}
\end{splitequation}
and the last equality holds under the stated assumption.
\end{proof}
The states defined in this way are called quasi-free (or Gaussian) with nonvanishing one-point function. Given a state, the GNS construction gives a Hilbert space $\mathcal{H}$, a cyclic representation $\pi$ of $\mathcal{A}$ such that $\pi(\mathcal{A}) \subseteq \mathcal{B}(\mathcal{H})$ and a cyclic vector $\Omega \in \mathcal{H}$ such that
\begin{equation}
\label{eq:phi_omega}
\phi\left( a \right) = \left( \Omega, \pi(a) \Omega \right)
\end{equation}
for all $a \in \mathcal{A}$. If the symplectic form $\sigma$ is nondegenerate, for a quasi-free state $\pi$ is the regular Fock representation over the one-particle Hilbert space $\mathcal{H}_1$ that is obtained by complexifying $H$ and completing it with respect to the norm $\norm{ f }_1 \defeq \sqrt{\left( f, f \right)_1} = \sqrt{\mu_2(f,f)}$.

On $\pi(\mathcal{A})$, we employ the strong operator topology, which is the one induced by the seminorms $\{ \norm{ \pi(a) \Psi } \colon \Psi \in \mathcal{H} \}$ for $a \in \mathcal{A}$. It is weaker than the operator norm topology, but makes the representations of the Weyl operators $\pi(W(f))$ continuous, and in particular allows to define the coordinate operators themselves:
\begin{definition}
Given a quasi-free state $\phi$ and the Fock representation $\pi(A)$, the coordinate operators $X(f)$ are defined by
\begin{equation}
\label{eq:xmu_def}
X(f) \Psi \defeq \lim_{t \to 0} \frac{\pi( W(t f) ) - \1}{\mathi t} \Psi
\end{equation}
for all those $\Psi \in \mathcal{H}$ for which the limit exists (in the norm topology on $\mathcal{H}$).
\end{definition}
Since the Fock representation is regular, the representations of the Weyl operators $\pi(W(t f))$ are continuous in $t$ in the strong operator topology, and thus form a strongly continuous one-parameter group $\bigl( \pi(W(t f)) \bigr)_{t \in \mathbb{R}}$ for any fixed $f$. By Stone's theorem, their generators $X(f)$ are self-adjoint operators defined on a dense domain in $\mathcal{H}$, which in particular includes the vector $\Omega$. However, they are unbounded (since otherwise it would follow that the $W(t f)$ are continuous in $t$ in the operator norm topology).

For a quasi-free state $\phi$ and its GNS representation $\pi$, it is then easy to see that
\begin{equations}[eq:coord_expectation_mu12]
\left( \Omega, X(f) \Omega \right) &= \mu_1(f) \eqcolon \phi\left( X(f) \right) \eqend{,} \\
\left( \Omega, X(f) X(g) \Omega \right) &= \left( f, g \right)_1 = \mu_2(f,g) + \frac{\mathi}{2} \sigma(f,g) \eqcolon \phi\left( X(f) X(g) \right)
\end{equations}
for $f, g \in H$.

Comparing with the results of our earlier work~\cite{froebmuchpapadopoulos2023a,froebmuchpapadopoulos2023b}, we thus would like to set
\begin{equations}[eq:mu12_def_proposal]
\mu_1(f) &= \int x^\mu f_{(\mu)}(x) \total^4 x \eqend{,} \\
\mu_2(f,f) &= \kappa^2 \iint f_{(\mu)}(x) f_{(\nu)}(x') \eta^{\mu\nu} \int F_{-+}(t,t',\vec{p}) \, \mathe^{\mathi \vec{p} \cdot (\vec{x}-\vec{x}')} \frac{\total^3 \vec{p}}{(2\pi)^3} \total^4 x \total^4 x' \eqend{,}
\end{equations}
where the correlation function $F_{-+}(t,t',\vec{p})$ is computed explicitly in \cite[Eq. (39)]{froebmuchpapadopoulos2023a}:
\begin{equation}
\label{eq:fmunu}
F_{-+}(t, t', \vec{p}) = \frac{1}{4 \abs{\vec{p}}^3} \, \mathe^{- \mathi \abs{\vec{p}} (t-t')} \left[ 1 + \mathi \abs{\vec{p}} (t-t') \right] \eqend{.}
\end{equation}
However, there are two problems with this choice. First of all, the inverse Fourier transform of $F$~\eqref{eq:fmunu} is infrared-divergent, since the function diverges $\sim \abs{\vec{p}}^{-3}$ for small $\abs{\vec{p}}$. Therefore, $\mu_2$~\eqref{eq:mu12_def_proposal} is finite only if the mean
\begin{equation}
\label{eq:def_fbar_mean}
\bar{f}_{(\mu)} \defeq \int f_{(\mu)}(x) \total^4 x
\end{equation}
vanishes. Second, even for functions with vanishing mean $\mu_2$ is not positive definite because of the indefinite signature of $\eta^{\mu\nu}$, and the fact that even the symmetrized $F_{-+}(t, t', \vec{p}) + F_{-+}(t', t, \vec{p})$ is not positive. Fortunately, both of these problems can be cured: instead of the above choice, we have to adopt an extension of the Dereziński--Meissner state~\cite{DerezinskiMeissner2006}, which deals with the infrared divergence and the fact that $F_{-+}(t, t', \vec{p}) + F_{-+}(t', t, \vec{p})$ is not positive, and the indefinite signature is resolved by employing a Krein space formulation.

\subsection{Krein space}
\label{subsec:krein}

Krein spaces arise in many situations in quantum field theory, such as quantum electrodynamics where they provide an extension of the earlier Gupta--Bleuler construction. We consider
\begin{definition}
A (complex) Krein space is a triple $(V, \left\langle \cdot, \cdot \right\rangle, J)$, where
\begin{itemize}
\item $V$ is a complex vector space,
\item $\left\langle \cdot, \cdot \right\rangle \colon V \times V \to \mathbb{C}$ is a Hermitian sesquilinear form (an indefinite inner product),
\item $J \colon V \to V$ (sometimes called a fundamental symmetry) is a linear map that is symmetric with respect to the indefinite inner product, i.e., fulfilling
\begin{equation}
\left\langle J f, g \right\rangle = \left\langle f, J g \right\rangle \quad \text{for all}\quad f, g \in V \eqend{,}
\end{equation}
\item the inner product $\left( \cdot, \cdot \right)$ defined by
\begin{equation}
\label{eq:krein_relation_innerproducts}
\left( f, g \right) \defeq \left\langle f, J g \right\rangle
\end{equation}
is positive definite and induces a majorant topology on $V$, and
\item $V$ is complete with respect to the corresponding norm.
\end{itemize}
\end{definition}
The weak topology on $V$ is the one defined by the family of seminorms $\{ \abs{ \left\langle f, g \right\rangle }\colon g \in V \}$ for all $f \in V$. It is the weakest partial majorant topology on $V$~\cite[Thm.~II.2.1]{bognar1974}, but unfortunately only a majorant topology if $\dim V < \infty$~\cite[Thm.~IV.1.4]{bognar1974}. The condition~\eqref{eq:krein_relation_innerproducts} is thus a true restriction on an indefinite inner product space.

We now use that any element $f \in V$ can be decomposed in positive and negative parts $f_\pm = (f \pm J f)/2$, for which we obtain
\begin{equation}
\pm \left\langle f_\pm, f_\pm \right\rangle = \left( f_\pm, f_\pm \right) \geq 0 \eqend{.}
\end{equation}
That is, we may decompose $V = V^+ + V^0 + V^-$ with $\left\langle f, f \right\rangle > 0$ for $f \in V^+$, $\left\langle f, f \right\rangle < 0$ for $f \in V^-$ and $\left\langle f, f \right\rangle = 0$ for $f \in V^0$. Since $\left( \cdot, \cdot \right)$ is positive definite, the set $V^0$ that consists of all $f \in V$ for which both $f_\pm = 0$, namely $V^0 = \operatorname{ker} J$, is actually empty. The inner product $\left\langle \cdot, \cdot \right\rangle$ is therefore nondegenerate, and $J$ is invertible, which together with the symmetry condition $\left\langle J f, g \right\rangle = \left\langle f, J g \right\rangle$ shows that $J^2 = \1$, i.e., $J$ is an involution. We thus have $V = V^+ + V^-$ (the fundamental decomposition), where each of these subspaces is closed in the majorant topology~\cite[Thm.~IV.5.2]{bognar1974}, and thus both $V^+$ and $V^-$ are Hilbert spaces. While the concrete decomposition depends on $J$, the corresponding majorant topologies on $V$ are all equivalent~\cite[Thm.~IV.6.4]{bognar1974}.

In our case, we choose an arbitrary timelike unit vector $u^\mu$ ($\eta_{\mu\nu} u^\mu u^\nu = -1$) and set
\begin{equation}
\label{eq:krein_j_def}
( J f )_{(\mu)}(x) \defeq \left( \delta_\mu^\nu + 2 u_\mu u^\nu \right) f_{(\nu)}(x) \eqend{,}
\end{equation}
such that $J^2 = \mathrm{id}$ and
\begin{equation}
f_{(\mu)}(x) \eta^{\mu\nu} ( J f )_{(\nu)}(x') = f_{(\mu)}(x) \left( \eta^{\mu\nu} + 2 u^\mu u^\nu \right) f_{(\rho)}(x') \eqend{.}
\end{equation}
Since for any $\alpha_\mu$ we have
\begin{splitequation}
\label{eq:krein_j_positive}
\alpha_\mu \left( \eta^{\mu\nu} + 2 u^\mu u^\nu \right) \alpha_\nu &= \left[ 1 + 2 \sum_{i=1}^3 ( u_i )^2 \right] ( \alpha_0 )^2 + \sum_{i=1}^3 ( \alpha_i )^2 + 2 \left( \sum_{i=1}^3 u^i \alpha_i \right)^2 \\
&\geq ( \alpha_0 )^2 + \sum_{i=1}^3 ( \alpha_i )^2 \geq \abs{ \alpha_\mu \eta^{\mu\nu} \alpha_\nu }
\end{splitequation}
the matrix $\eta^{\mu\nu} + 2 u^\mu u^\nu$ is positive definite and dominates $\eta_{\mu\nu}$. Hence, it induces a majorant topology on $V$, and completing $V$ with respect to the scalar product~\eqref{eq:krein_relation_innerproducts} yields the desired Krein space structure, which we construct explicitly in the next subsection.

\subsection{Dereziński--Meissner state}
\label{subsec:dm_state}

The infrared divergence appearing in the inverse Fourier transform of $F^{\mu\nu}$~\eqref{eq:fmunu} is analogous to the one for a massless scalar field in two dimensions, which (among many others, see for example~\cite{nakanishi1980,morchiopierottistrocchi1990}) was studied and finally resolved by Dereziński and Meissner~\cite{DerezinskiMeissner2006}. We generalize their arguments to our situation, analogously to the construction of Schubert~\cite{schubert2011}. For this, we first compute the Fourier transform in~\eqref{eq:mu12_def_proposal} for functions with vanishing mean:
\begin{lemma}
For a function $f \in C_0^\infty(\mathbb{R}^4)$ with $\bar{f} = 0$, we have
\begin{splitequation}
&\iint f(x) f(x') \int \frac{1}{4 \abs{\vec{p}}^3} \, \mathe^{- \mathi \abs{\vec{p}} (t-t')} \left[ 1 + \mathi \abs{\vec{p}} (t-t') \right] \, \mathe^{\mathi \vec{p} \cdot (\vec{x}-\vec{x}')} \frac{\total^3 \vec{p}}{(2\pi)^3} \total^4 x \total^4 x' \\
&= - \frac{1}{16 \pi^2} \iint f(x) f(x') \ln\abs{ (x-x')^2 } \total^4 x \total^4 x' \eqend{.}
\end{splitequation}
\end{lemma}
\begin{proof}
The inverse Fourier transform is a priori defined in the sense of distributions, namely
\begin{splitequation}
&\iint f(x) f(x') \int \frac{1}{4 \abs{\vec{p}}^3} \, \mathe^{- \mathi \abs{\vec{p}} (t-t')} \left[ 1 + \mathi \abs{\vec{p}} (t-t') \right] \, \mathe^{\mathi \vec{p} \cdot (\vec{x}-\vec{x}')} \frac{\total^3 \vec{p}}{(2\pi)^3} \total^4 x \total^4 x' \\
&\quad= \int \frac{1}{4 \abs{\vec{p}}^3} \left[ \Bigl( 1 + \abs{\vec{p}} \partial_{k^0} - \abs{\vec{p}} \partial_{q^0} \Bigr) ( \mathcal{F} f )(k^0,-\vec{p}) ( \mathcal{F} g )(q^0,\vec{p}) \right]_{q^0 = - k^0 = \abs{\vec{p}}} \frac{\total^3 \vec{p}}{(2\pi)^3}
\end{splitequation}
with the Fourier transform defined by
\begin{equation}
( \mathcal{F} f )(p^0,\vec{p}) \defeq \int f(t,\vec{x}) \, \mathe^{\mathi p^0 t - \mathi \vec{p} \cdot \vec{x}} \total t \total^3 \vec{x} \eqend{.}
\end{equation}
Since $f$ has vanishing mean, its Fourier transform vanishes at the origin, and since it is smooth the Fourier transforms decays faster than any power of $\abs{\vec{p}}$, hence the integral over $\vec{p}$ is well defined in the ordinary sense. We may thus compute the original Fourier transform by introducing a cutoff $\delta$ for small $\abs{\vec{p}}$ and a factor $\mathe^{- \epsilon \abs{\vec{p}}}$ which serves as a cutoff for large $\abs{\vec{p}}$, and take the limit $\delta,\epsilon \to 0^+$ at the end. It follows that
\begin{splitequation}
&\int \frac{1}{4 \abs{\vec{p}}^3} \, \mathe^{- \mathi \abs{\vec{p}} t} \left( 1 + \mathi \abs{\vec{p}} t \right) \, \mathe^{\mathi \vec{p} \cdot \vec{x}} \frac{\total^3 \vec{p}}{(2\pi)^3} \\
&\quad= \lim_{\delta,\epsilon \to 0^+} \int_{\abs{\vec{p}} \geq \delta} \frac{1}{4 \abs{\vec{p}}^3} \, \mathe^{- \mathi \abs{\vec{p}} t} \left( 1 + \mathi \abs{\vec{p}} t \right) \, \mathe^{\mathi \vec{p} \cdot \vec{x}} \mathe^{- \epsilon \abs{\vec{p}}} \frac{\total^3 \vec{p}}{(2\pi)^3} \\
&\quad= \frac{1}{8 \pi^2 \abs{\vec{x}}} \lim_{\delta,\epsilon \to 0^+} \int_\delta^\infty \, \mathe^{- \mathi p t} \left( 1 + \mathi p t \right) \, \frac{\sin(p \abs{\vec{x}})}{p^2} \mathe^{- \epsilon p} \total p \\
&\quad= \frac{1}{8 \pi^2 \abs{\vec{x}}} \lim_{\delta,\epsilon \to 0^+} \left[ \frac{- \abs{\vec{x}} + \mathi \epsilon}{2} \operatorname{E}_1\left[ \epsilon + \mathi p (\abs{\vec{x}}+t) \right] - \frac{\abs{\vec{x}} + \mathi \epsilon}{2} \operatorname{E}_1\left[ \epsilon - \mathi p (\abs{\vec{x}}-t) \right] - \frac{\sin(p \abs{\vec{x}})}{p} \mathe^{- \mathi p t - \epsilon p} \right]_\delta^\infty \\
&\quad= \frac{1}{16 \pi^2} \lim_{\delta,\epsilon \to 0^+} \left[ 2 - 2 \gamma - 2 \ln \delta - \left( 1 + \mathi \frac{\epsilon}{\abs{\vec{x}}} \right) \ln\left[ \epsilon - \mathi (\abs{\vec{x}}-t) \right] - \left( 1 - \mathi \frac{\epsilon}{\abs{\vec{x}}} \right) \ln\left[ \epsilon + \mathi (\abs{\vec{x}}+t) \right] \right] \eqend{,}
\end{splitequation}
where $\operatorname{E}_n(z)$ is the generalized exponential integral of order $n$ and $\gamma$ is the Euler--Mascheroni constant. To compute the limit, we note that constants do not contribute because $f$ has vanishing mean, while for the logarithms we use $\ln z = \ln \abs{z} + \mathi \arg z$ to obtain
\begin{splitequation}
\label{eq:f_log_computation}
&\iint f(x) f(x') \int \frac{1}{4 \abs{\vec{p}}^3} \, \mathe^{- \mathi \abs{\vec{p}} (t-t')} \Bigl[ 1 + \mathi \abs{\vec{p}} (t-t') \Bigr] \, \mathe^{\mathi \vec{p} \cdot (\vec{x}-\vec{x}')} \frac{\total^3 \vec{p}}{(2\pi)^3} \total^4 x \total^4 x' \\
&\quad= - \frac{1}{16 \pi^2} \lim_{\epsilon \to 0^+} \iint f(x) f(x') \Bigl[ \ln\left[ \epsilon - \mathi \left( \abs{\vec{x}-\vec{x}'}-t+t' \right) \right] + \ln\left[ \epsilon + \mathi \left( \abs{\vec{x}-\vec{x}'}+t-t' \right) \right] \Bigr] \total^4 x \total^4 x' \\
&\quad= - \frac{1}{16 \pi^2} \iint f(x) f(x') \Bigl[ \ln\abs{ \abs{\vec{x}-\vec{x}'}-t+t' } + \ln\abs{ \abs{\vec{x}-\vec{x}'}+t-t' } \\
&\qquad\qquad- \frac{\mathi \pi}{2} \sgn\left( \abs{\vec{x}-\vec{x}'}-t+t' \right) + \frac{\mathi \pi}{2} \sgn\left( \abs{\vec{x}-\vec{x}'}+t-t' \right) \Bigr] \total^4 x \total^4 x' \\
&\quad= - \frac{1}{16 \pi^2} \iint f(x) f(x') \ln\abs{ \abs{\vec{x}-\vec{x}'}^2 - (t-t')^2 } \total^4 x \total^4 x' \eqend{,}
\end{splitequation}
where in the third equality we exchanged $x$ and $x'$ in the last term in the integrand, which then cancelled with the third one, and combined the remaining logarithms.
\end{proof}
Unfortunately, in contrast to the two-dimensional case, where $\ln \abs{ (x-x')^2 }$ is a conditionally negative distribution such that $\iint f(x) \ln \abs{ (x-x')^2 } f(x') \total^2 x \total^2 x' \leq 0$ for all functions $f$ with vanishing mean~\cite{DerezinskiMeissner2006,schubert2011}, in four dimensions it is indefinite if the support of $f$ is large. In addition to the construction of Dereziński and Meissner, we thus have to cut off the positive part of $\ln \abs{ (x-x')^2 }$ and define
\begin{definition}[Dereziński--Meissner state]
\label{def:dmstate}
Given a test function $\psi \in C_0^\infty(\mathbb{R}^4)$ with mean 1 (but otherwise arbitrary) and $\alpha \in (0,\infty)$, consider the bilinear form
\begin{splitequation}
\label{eq:dmstate_tau}
\Delta_{\alpha,\psi}(f,g) &\defeq - \frac{\kappa^2}{16 \pi^2} \iint \ln\abs{ (x-x')^2 }_- ( P_\psi f )_\mu(x) \eta^{\mu\nu} ( P_\psi g )_\nu(x') \total^4 x \total^4 x' \\
&\quad+ \frac{\mathi}{2} \sigma(P_\psi f, P_\psi g) + \alpha \kappa^2 \left( \bar{f}_\mu + \frac{\mathi}{2 \alpha \kappa^2} \sigma(f,\psi)_\mu \right) \eta^{\mu\nu} \left( \bar{g}_\nu - \frac{\mathi}{2 \alpha \kappa^2} \sigma(g,\psi)_\nu \right) \\
&= - \frac{\kappa^2}{16 \pi^2} \iint \ln\abs{ (x-x')^2 }_- ( P_\psi f )_\mu(x) \eta^{\mu\nu} ( P_\psi g )_\nu(x') \total^4 x \total^4 x' \\
&\quad+ \alpha \kappa^2 \, \bar{f}_\mu\, \eta^{\mu\nu} \bar{g}_\nu + \frac{1}{4 \alpha \kappa^2} \sigma(f,\psi)_\mu\, \eta^{\mu\nu} \sigma(g,\psi)_\nu + \frac{\mathi}{2} \sigma(f, g) \eqend{,}
\end{splitequation}
where the projection operator $P_\psi$ on functions with vanishing mean is defined by
\begin{equation}
\label{eq:dmstate_proj}
(P_\psi f)_{(\mu)}(x) \defeq f_{(\mu)}(x) - \bar{f}_{(\mu)} \, \psi(x) \eqend{,}
\end{equation}
and the bidistribution whose formal integral kernel we denote by $\ln\abs{ (x-x')^2 }_-$ is defined by
\begin{splitequation}
\label{eq:ln_abs_x2_minus_def}
&\iint f(x) \ln\abs{ (x-x')^2 }_- g(x') \total^4 x \total^4 x' \\
&\quad\defeq \frac{1}{4} \min\left[ \iint [ f(x) + g(x) ] \ln\abs{ (x-x')^2 } [ f(x') + g(x') ] \total^4 x \total^4 x', 0 \right] \\
&\qquad- \frac{1}{4} \min\left[ \iint [ f(x) - g(x) ] \ln\abs{ (x-x')^2 } [ f(x') - g(x') ] \total^4 x \total^4 x', 0 \right] \eqend{.}
\end{splitequation}
The Dereziński--Meissner state is the quasi-free state $\omega$ on the Weyl algebra $\mathcal{A}$ with scalar product
\begin{equation}
\label{eq:dmstate_scalarprod}
\left( f , g \right)_1 \defeq \Delta_{\alpha,\psi}\left( f, J g \right)
\end{equation}
on the one-particle Hilbert space $\mathcal{H}_1$ of the Fock representation of $\mathcal{A}$, which depends on the Krein space involution $J$~\eqref{eq:krein_j_def}.
\end{definition}
This definition makes sense if $\left( f , f \right)_1 \geq 0$, which is clear from the form of $\Delta_{\alpha,\psi}$~\eqref{eq:dmstate_tau} and the positive definiteness of $J$~\eqref{eq:krein_j_positive} if the first term is positive. That this holds follows in turn from the definition~\eqref{eq:ln_abs_x2_minus_def}. Comparing the scalar product~\eqref{eq:dmstate_scalarprod} with the expression~\eqref{eq:hilbert_1_scalarprod}, we see that we have replaced the ill-defined, nonpositive functional $\mu_2$ of~\eqref{eq:mu12_def_proposal} by a well-defined and positive one, namely we define
\begin{equations}[eq:mu12_def]
\mu_1(f) &\defeq \int x^\mu f_{(\mu)}(x) \total^4 x \eqend{,} \\
\mu_2(f,f) &\defeq \Delta_{\alpha,\psi}\left( f, J f \right) \eqend{,}
\end{equations}
and change the symplectic form $\sigma \to \sigma( \cdot, J \cdot )$, with $J$ defined by~\eqref{eq:krein_j_def}. In a sense, this is the minimal change that we have to make to obtain a quantum theory of the noncommutative coordinates $X^{(\mu)}$. We note that the cutoff imposed by the bidistribution $\ln\abs{ (x-x')^2 }_-$ becomes irrelevant if $f$ has sufficiently small support:
\begin{lemma}
\label{lemma:minimal_variance}
Let $f(x) = \chi_p(x) - \chi_q(x)$, where
\begin{equation}
\chi_p(x) = \left( \frac{\alpha}{\pi} \right)^2 \exp\left[ - \alpha \sum_{\mu=1}^4 (x^\mu-p^\mu)^2 \right]
\end{equation}
is a Gaussian with variance $2/\alpha$. In the limit $\alpha \to \infty$, we have
\begin{equations}[eq:lemma_minimal_variance]
&\lim_{\alpha \to \infty} \left[ \iint f(x) \ln\abs{ (x-x')^2 }_- \, f(x') \total^4 x \total^4 x' - \iint f(x) \ln\abs{ (x-x')^2 } f(x') \total^4 x \total^4 x' \right] = 0 \eqend{,} \\
&\lim_{\alpha \to \infty} \left[ \iint f(x) \ln\abs{ (x-x')^2 } f(x') \total^4 x \total^4 x' + 2 \ln\abs{ \alpha (p-q)^2 } \right] = 4 (1-\gamma) \eqend{.}
\end{equations}
\end{lemma}
\begin{proof}
Since
\begin{equation}
\int \chi_p(x) \total^4 x = \left[ \sqrt{ \frac{\alpha}{\pi} } \int \mathe^{- \alpha x^2} \total x \right]^4 = 1 \eqend{,}
\end{equation}
we have $\lim_{\alpha \to \infty} \chi_p(x) = \delta^4(x-p)$, and it follows that
\begin{splitequation}
&\lim_{\alpha \to \infty} \left[ \iint f(x) \ln\abs{ (x-x')^2 } f(x') \total^4 x \total^4 x' + 2 \ln\abs{ \alpha (p-q)^2 } \right] \\
&\quad= 2 \lim_{\alpha \to \infty} \left( \frac{\alpha}{\pi} \right)^4 \iint \exp\left[ - \alpha \sum_{\mu=1}^4 (x^\mu)^2 - \alpha \sum_{\mu=1}^4 ((x')^\mu)^2 \right] \ln\left[ \alpha \abs{ (x-x')^2 } \right] \total^4 x \total^4 x' \\
&\quad= \frac{2}{\pi^4} \iint \exp\left[ - 2 \sum_{\mu=1}^4 (z^\mu)^2 - \frac{1}{2} \sum_{\mu=1}^4 (y^\mu)^2 \right] \ln\abs{ y^2 } \total^4 y \total^4 z \eqend{,}
\end{splitequation}
where we made the change of variables $y^\mu = \sqrt{\alpha} (x-x')^\mu$, $z^\mu = \frac{1}{2} \sqrt{\alpha} (x+x')^\mu$. Separating time and space and passing to spherical coordinates, this gives
\begin{splitequation}
&\lim_{\alpha \to \infty} \left[ \iint f(x) \ln\abs{ (x-x')^2 } f(x') \total^4 x \total^4 x' + 2 \ln\abs{ \alpha (p-q)^2 } \right] \\
&\quad= \frac{2}{\pi} \int_0^\infty \int_0^\infty \exp\left( - \frac{1}{2} t^2 - \frac{1}{2} r^2 \right) \ln\abs{ t^2 - r^2 } r^2 \total r \total t \\
&\quad= \frac{1}{8 \pi} \int_{-\infty}^\infty \int_{-\infty}^\infty \exp\left( - \frac{1}{4} u^2 - \frac{1}{4} v^2 \right) (u-v)^2 \left( \ln u^2 + \ln v^2 \right) \total u \total v \\
&\quad= \frac{1}{2 \sqrt{\pi}} \int_{-\infty}^\infty \exp\left( - \frac{1}{4} u^2 \right) (2+u^2) \ln u^2 \total u = 4 (1-\gamma) \eqend{,}
\end{splitequation}
where we made another change of variables $u = t-r$, $v = t+r$ and where $\gamma$ is the Euler--Mascheroni constant. Therefore, for sufficiently large $\alpha$ we have
\begin{equation}
\iint f(x) \ln\abs{ (x-x')^2 } f(x') \total^4 x \total^4 x' \approx - 2 \ln\abs{ \alpha (p-q)^2 } + 4 (1-\gamma) \eqend{,}
\end{equation}
which is negative such that the cutoff in the definition~\eqref{eq:ln_abs_x2_minus_def} of $\ln \abs{ (x-x')^2 }_-$ does not enter, and hence the result \eqref{eq:lemma_minimal_variance} follows.
\end{proof}

Recapitulating, we began with an algebra of test functions $f_{(\mu)}$ defined on a real symplectic space $(H,\sigma)$. On this space, we can consider any suitable continuous topology, for example the one induced by an $L_p$ or a Sobolev norm. Quantizing this symplectic space means the construction of the Weyl algebra $\mathcal{A}$, whose natural topology (the operator norm topology induced by the $C^*$ norm) is however uniformly discrete, and thus too fine and unsuitable to define noncommutative coordinates as operators. Given a quasi-free state $\phi$ on $\mathcal{A}$, we therefore employed the GNS construction to obtain the Fock space representation $\pi(\mathcal{A})$. In this representation, the coarser strong operator topology allowed us to define the coordinate operators $X^{(\mu)}(f)$, which are indexed by the completion of the algebra of test functions that we started with. In turn, to define the quasi-free state and ensure the positivity condition, we needed to adopt the analog of the Dereziński--Meissner state and employ a Krein space formulation.

\section{The noncommutative spacetime}
\label{sec:nc_spacetime}

Given the quantum theory, we now have to extract the structure of our noncommutative spacetime. Our approach adopts a constructive perspective similar to the Ehlers--Pirani--Schild (EPS) framework~\cite{ehlerspiranischild1972}. The EPS construction provides a foundational axiomatization of spacetime geometry that differs fundamentally from conventional approaches in general relativity. Rather than postulating the metric tensor as a primitive geometric object, the EPS method systematically reconstructs the spacetime structure from empirically accessible observational data: collections of events, light ray trajectories, and worldlines of freely falling particles. From these, a conformal structure and a projective structure are derived and shown to be compatible, which in turn results in a unique affine connection. Finally, the proper time along worldlines of particles and the metric can be obtained from the equations of parallel transport, or experimentally realized by passing Radar signals between neighboring particles.\footnote{We note that it is possible to define relational observables in perturbative quantum gravity in this way~\cite{khavkine2012,bongakhavkine2014}.}

Since our quantization of the noncommutative coordinates~\eqref{eq:coord_commutator} does not involve light rays or particles, we have to be content with somewhat more primitive data. Namely, we intend to extract from the Fock representation $\pi(\mathcal{A})$ of the Weyl algebra and the Dereziński--Meissner state $\omega$ a distance between points and their causal relation. Consider first the classical theory, which is obtained in the limit $\hbar \to 0$, i.e. for $\kappa = 0$. Since the symplectic form $\sigma$~\eqref{eq:sigma_def} vanishes for $\kappa = 0$, in this limit the Weyl algebra becomes commutative, and since $\mu_2$~\eqref{eq:mu12_def} also vanishes, only $\mu_1$ makes a nonvanishing contribution. In particular, from equations~\eqref{eq:phi_wf}, \eqref{eq:dmstate_scalarprod} and \eqref{eq:mu12_def} we obtain
\begin{equation}
\lim_{\hbar \to 0} \omega\left( W(f) \right) = \mathe^{\mathi \mu_1(f)} = \exp\left[ \mathi \int x^\mu f_{(\mu)}(x) \total^4 x \right] \eqend{.}
\end{equation}
Similar to the evaluation functionals $p \colon f \mapsto f(p)$ that are employed to recover the points of a Riemannian manifold from the canonically associated commutative spectral triple, we thus see that we need to consider certain singular $f_{(\mu)}$ to extract the coordinates of individual points from our construction. Namely, the formal limit $f_{(\mu)}(x) \to v_\mu \delta^4(x-p)$ (the localization limit) results in
\begin{equation}
\exp\left[ \mathi \int x^\mu f_{(\mu)}(x) \total^4 x \right] \to \exp\left( \mathi v_\mu p^\mu \right) \eqend{,}
\end{equation}
and by considering the one-parameter family $W(t f)$ and the limit~\eqref{eq:xmu_def} to obtain the coordinate operators themselves, we have
\begin{equation}
\lim_{\hbar \to 0} \omega\left( X(f) \right) \to v_\mu p^\mu \eqend{.}
\end{equation}
That is, we may extract the coordinate of the point $p$ in the direction $v$.

\subsection{Distance functional}
\label{subsec:distance}

To obtain a distance functional, we take the expectation of differences of coordinates. Again in the classical limit $\hbar \to 0$, we compute
\begin{equation}
\lim_{\hbar \to 0} \omega\left( [ X(f) - X(g) ]^2 \right) = \iint x^\mu y^\nu [ f_{(\mu)}(x) - g_{(\mu)}(x) ] [ f_{(\nu)}(y) - g_{(\nu)}(y) ] \total^4 x \total^4 y \eqend{,}
\end{equation}
and thus in the formal limit $f_{(\mu)}(x) \to v_\mu \delta^4(x-p)$ and $g_{(\mu)}(x) \to w_\mu \delta^4(x-q)$ we obtain
\begin{equation}
\lim_{\hbar \to 0} \omega\left( [ X(f) - X(g) ]^2 \right) \to \left( v_\mu p^\mu - w_\mu q^\mu \right)^2 \eqend{.}
\end{equation}
As expected, in the classical limit there are no cross-correlations, and we simply have $\lim_{\hbar \to 0} \omega\left( [ X(f) - X(g) ]^2 \right) = \lim_{\hbar \to 0} \left[ \omega\left( X(f) \right) - \omega\left( X(g) \right) \right]^2$. If we now sum over an orthonormal Lorentzian frame $v_\mu = w_\mu = e_\mu^{(a)}$ with $a = 0,\ldots,3$ and $e_\mu^{(a)} = \delta_\mu^a$, we recover the Minkowski distance
\begin{equation}
\sum_{a=1}^4 \eta_{aa} \left( e_\mu^{(a)} p^\mu - e_\mu^{(a)} q^\mu \right)^2 = - \left( p^0 - q^0 \right)^2 + \sum_{i=1}^3 \left( p^i - q^i \right)^2 = (p-q)^2 = 2 \sigma(p, q) \eqend{,}
\end{equation}
where $\sigma$ is the Synge world function~\cite{synge1960}.

Our viewpoint is thus somewhat dual to the one of spectral triples explained in the introduction, where points are identified with evaluation functionals and the distance is obtained as a supremum over a subset of the algebra~\eqref{eq:spectral_triple_distance}. In contrast, we have a fixed functional $\omega$ on the algebra (the Dereziński--Meissner state), and the points are given by the expectation values of the noncommutative coordinate operators in the localization limit.

Consider now the quantum theory with $\kappa \neq 0$. Considering the one-parameter family $W(t f)$ and the limit~\eqref{eq:xmu_def} to obtain the coordinate operators themselves, we compute
\begin{splitequation}
\omega\left( [ X(f) - X(g) ]^2 \right) &= - \lim_{s, t \to 0} \frac{1}{s t} \omega\left( \left[ W(s f) - W(s g) \right] \left[ W(t f) - W(t g) \right] \right) \\
&= - \lim_{s, t \to 0} \frac{1}{s t} \omega\Bigl( W((s+t) f) + W((s+t) g) \\
&\hspace{4em}- W(s f+t g) \mathe^{- \frac{\mathi}{2} s t \sigma(f, J g)} - W(s g+t f) \mathe^{- \frac{\mathi}{2} s t \sigma(g, J f)} \Bigr) \\
&= - \lim_{s, t \to 0} \frac{1}{s t} \biggl[ \exp\left[ \mathi (s+t) \mu_1(f) - \frac{1}{2} (s+t)^2 \mu_2(f,f) \right] \\
&\hspace{4em}+ \exp\left[ \mathi (s+t) \mu_1(g) - \frac{1}{2} (s+t)^2 \mu_2(g,g) \right] \\
&\hspace{4em}- \exp\left[ \mathi \mu_1(s f+t g) - \frac{1}{2} \mu_2(s f+t g,s f+t g) \right] \mathe^{- \frac{\mathi}{2} s t \sigma(f, J g)} \\
&\hspace{4em}- \exp\left[ \mathi \mu_1(s g+t f) - \frac{1}{2} \mu_2(s g+t f,s g+t f) \right] \mathe^{- \frac{\mathi}{2} s t \sigma(g, J f)} \biggr] \\
&= [ \mu_1(f) - \mu_1(g) ]^2 + \mu_2(f-g,f-g) - \frac{\mathi}{2} \sigma(f, J g) - \frac{\mathi}{2} \sigma(g, J f) \\
&= \iint x^\mu y^\nu [ f_{(\mu)}(x) - g_{(\mu)}(x) ] [ f_{(\nu)}(y) - g_{(\nu)}(y) ] \total^4 x \total^4 y \\
&\quad+ \Delta_{\alpha,\psi}\left( f-g, J (f-g) \right) \eqend{,}
\end{splitequation}
where in the second equality we used the Weyl relations in the Krein space. In the last equality we used the replacements~\eqref{eq:mu12_def} for the Dereziński--Meissner state, as well as the relations
\begin{equation}
\mathi \sigma(f, g) = \Delta_{\alpha,\psi}\left( f, g \right) - \Delta_{\alpha,\psi}\left( g, f \right) \eqend{,} \quad \Delta_{\alpha,\psi}\left( f, J g \right) = \Delta_{\alpha,\psi}\left( J f, g \right) \eqend{,}
\end{equation}
which follow from its definition~\eqref{eq:dmstate_tau} and the one of $J$~\eqref{eq:krein_j_def}.

In the formal localization limit $f_{(\mu)}(x) \to v_\mu \delta^4(x-p)$ and $g_{(\mu)}(x) \to w_\mu \delta^4(x-q)$, this gives
\begin{splitequation}
&\omega\left( [ X(f) - X(g) ]^2 \right) \to \left( v_\mu p^\mu - w_\mu q^\mu \right)^2 + \alpha \kappa^2 ( v_\mu - w_\mu ) \left( \eta^{\mu\nu} + 2 u^\mu u^\nu \right) ( v_\nu - w_\nu ) \\
&\quad+ \frac{\kappa^2}{16 \pi^2} \iint \ln\abs{ (x-x')^2 } \left( v_\mu [ \delta^4(x-p) - \psi(x) ] - w_\mu [ \delta^4(x-q) - \psi(x) ] \right) \left( \eta^{\mu\nu} + 2 u^\mu u^\nu \right) \\
&\qquad\quad\times \left( v_\nu [ \delta^4(x'-p) - \psi(x') ] - w_\nu [ \delta^4(x'-q) - \psi(x') ] \right) \total^4 x \total^4 x' \\
&\quad+ \frac{\kappa^2}{2^8 \alpha \pi^2} \iint \left[ v_\mu \sgn(p^0-t) \Theta\left[ -(p-x)^2 \right] - w_\mu \sgn(q^0-t) \Theta\left[ -(q-x)^2 \right] \right] \left( \eta^{\mu\nu} + 2 u^\mu u^\nu \right) \\
&\qquad\quad\times \left[ v_\nu \delta^4(x-p) \sgn(p^0-t') \Theta\left[ -(p-x')^2 \right] - w_\nu \delta^4(x-q) \sgn(q^0-t') \Theta\left[ -(q-x')^2 \right] \right] \\
&\qquad\quad\times \psi(x) \psi(x') \total^4 x \total^4 x' \eqend{,}
\end{splitequation}
where we used the explicit expressions for the Dereziński--Meissner state~\eqref{eq:dmstate_tau}, the projector $P_\psi$~\eqref{eq:dmstate_proj}, the symplectic form $\sigma$~\eqref{eq:sigma_def} and the involution $J$~\eqref{eq:krein_j_def}. However, this expression suffers from two problems: first, it contains the arbitrary vector $u^\mu$ that we use to define the involution and is thus not Lorentz-covariant, and second, it contains divergent terms such as $\ln\abs{ (p - p)^2 }$. The second problem is a consequence of the unboundedness of the coordinate operators, and tells us that we cannot strictly localize points in the quantum theory. On the other hand, the first problem can be solved by employing instead of $\omega$ the nonpositive linear functional $\tau_{\alpha,\psi}$, defined on the Weyl algebra $\mathrm{CCR}(H,\sigma)$ by
\begin{equation}
\label{eq:tau_weyl_def}
\tau_{\alpha,\psi}(W(f)) = \exp\left[ \mathi \int x^\mu f_{(\mu)}(x) \total^4 x - \frac{1}{2} \Delta_{\alpha,\psi}\left( f, f \right) \right] \eqend{.}
\end{equation}
Since we introduced the Krein space only to obtain a positive state, such that it is only an auxiliary construct, this is in fact the physically sensible choice. We thus define
\begin{definition}
Given positive test functions $\chi_p, \chi_q \in \mathbb{C}_0^\infty(\mathbb{R}^4)$ with mean $1$, the distance between them is defined as
\begin{splitequation}
\mathcal{D}_{\alpha,\psi}(\chi_p,\chi_q) &\defeq \eta_{ab} \, \tau_{\alpha,\psi}\left( \left[ X(f^{(a)}) - X(g^{(a)}) \right] \left[ X(f^{(b)}) - X(g^{(b)}) \right] \right) \\
&= \eta_{ab} \iint x^\mu (x')^\nu \left[ f^{(a)}_\mu(x) - g^{(a)}_\mu(x) \right] \left[ f^{(b)}_\nu(x') - g^{(b)}_\nu(x') \right] \total^4 x \total^4 x' \\
&\quad+ \eta_{ab} \, \Delta_{\alpha,\psi}\left( f^{(a)} - g^{(a)}, f^{(b)} - g^{(b)} \right)
\end{splitequation}
with
\begin{equation}
f^{(a)}_\mu(x) = e_\mu^{(a)} \chi_p(x) \eqend{,} \quad g^{(a)}_\mu(x) = e_\mu^{(a)} \chi_q(x) \eqend{,}
\end{equation}
and where $\{ e_\mu^{(a)} \}_{a=1}^4$ is an orthonormal Lorentzian frame.
\end{definition}

We interpret the test functions $\chi_p$ and $\chi_q$ as a measure of the precision with which the points $p$ and $q$ are determined, for example a Gaussian with a small variance. These will depend on the experiment, but as we have seen above, the idealized limit $\chi_p(x) \to \delta^4(x-p)$ where the localization becomes exact is divergent. However, the classical part of the distance functional $\mathcal{D}$ can in fact be exactly localized and then give back the usual Lorentzian distance. On the other hand, the test function $\psi$ and the constant $\alpha$ that enter the Dereziński--Meissner state are a property of the state and do not depend on the experiment. Nevertheless, we may take the limit $\alpha \to \infty$ of the distance function:
\begin{lemma}
The limit $\lim_{\alpha \to \infty} \mathcal{D}_{\alpha,\psi}(\chi_p,\chi_q)$ exists and is independent of $\psi$.
\end{lemma}
\begin{proof}
Since the test functions $\chi_p$ and $\chi_q$ have mean 1, we obtain $\left[ P_\psi\left( f^{(a)} - g^{(a)} \right) \right]_\mu(x) = f^{(a)}_\mu(x) - g^{(a)}_\mu(x)$ and thus
\begin{splitequation}
\mathcal{D}_{\alpha,\psi}(\chi_p,\chi_q) &= \eta_{ab} \iint x^\mu (x')^\nu \left[ f^{(a)}_\mu(x) - g^{(a)}_\mu(x) \right] \left[ f^{(b)}_\nu(x') - g^{(b)}_\nu(x') \right] \total^4 x \total^4 x' \\
&\quad- \frac{\kappa^2}{16 \pi^2} \eta_{ab} \iint \ln\abs{ (x-x')^2 }_- \left[ f^{(a)}_\mu(x) - g^{(a)}_\mu(x) \right] \eta^{\mu\nu} \left[ f^{(b)}_\nu(x') - g^{(b)}_\nu(x') \right] \total^4 x \total^4 x' \\
&\quad+ \frac{1}{4 \alpha \kappa^2} \eta_{ab} \sigma\left( f^{(a)} - g^{(a)}, \psi \right)_\mu\, \eta^{\mu\nu} \sigma\left( f^{(b)} - g^{(b)}, \psi \right)_\nu \eqend{,}
\end{splitequation}
and in the limit $\alpha \to \infty$ the last term vanishes.
\end{proof}

We thus obtain a Lorentzian noncommutative distance functional:
\begin{definition}
Given positive test functions $\chi_p, \chi_q \in \mathbb{C}_0^\infty(\mathbb{R}^4)$ with mean $1$, the (Lorentzian) noncommutative distance between them is defined as
\begin{splitequation}
\mathcal{D}(\chi_p, \chi_q) &\defeq \lim_{\alpha \to \infty} \mathcal{D}_{\alpha,\psi}(\chi_p,\chi_q) \\
&= \eta_{\mu\nu} \iint x^\mu (x')^\nu \left[ \chi_p(x) - \chi_q(x) \right] \left[ \chi_p(x') - \chi_q(x') \right] \total^4 x \total^4 x' \\
&\quad- \frac{\kappa^2}{4 \pi^2} \iint \ln\abs{ (x-x')^2 }_- \left[ \chi_p(x) - \chi_q(x) \right] \left[ \chi_p(x') - \chi_q(x') \right] \total^4 x \total^4 x' \eqend{.}
\end{splitequation}
\end{definition}
To obtain the second equality, we used that $\eta_{ab} e_\mu^{(a)} e_\nu^{(b)} \eta^{\mu\nu} = 4$. This distance functional reduces to the classical Minkowski distance $(p-q)^2$ or Synge world function~\cite{synge1960} $\sigma(p,q) = \frac{1}{2} (p-q)^2$ as $\hbar \to 0$ (which engenders $\kappa \to 0$) and in the localization limit $\chi_p(x) \to \delta^4(x-p)$, but in the quantum theory is supplemented by corrections of order $\kappa^2 \sim \ell_\mathrm{Pl}^2$, reflecting the underlying noncommutative structure. Moreover, while the classical Minkowski distance can have either sign, the quantum correction is always positive because $\ln\abs{ (x-x')^2 }_-$ is a negative semidefinite kernel~\eqref{eq:ln_abs_x2_minus_def}. Our result is thus reminiscent of a minimal or zero-point length of the form proposed in~\cite{padmanabhan1987,padmanabhan1997,kothawala2013,kothawala2023}, namely $(p-q)^2 \to (p-q)^2 + L_0^2$ for some fixed length scale $L_0$. However, in our proposal the minimal length depends on the points $p$ and $q$, as well as on the precision with which they are determined, and is thus dynamical. We may estimate it using Lemma~\ref{lemma:minimal_variance}, choosing a Gaussian
\begin{equation}
\chi_p(x) = \left( \frac{2 c}{\pi \ell_\mathrm{Pl}^2} \right)^2 \exp\left( - 2 c \sum_{\mu=1}^4 \frac{(x^\mu-p^\mu)^2}{\ell_\mathrm{Pl}^2} \right)
\end{equation}
with variance $\ell_\mathrm{Pl}^2/c$ proportional to the squared Planck length. In the limit $\ell_\mathrm{Pl} \to 0$, we obtain from equation \eqref{eq:lemma_minimal_variance}
\begin{splitequation}
\iint \ln\abs{ (x-x')^2 }_- \left[ \chi_p(x) - \chi_q(x) \right] \left[ \chi_p(x') - \chi_q(x') \right] \total^4 x \total^4 x' \to - 2 \ln\left[ \frac{2 c}{\mathe^{2 (1-\gamma)}} \frac{\abs{ (p-q)^2 }}{\ell_\mathrm{Pl}^2} \right] \eqend{,}
\end{splitequation}
and the distance functional including the leading quantum corrections thus reads
\begin{equation}
\mathcal{D}(\chi_p, \chi_q) \approx (p-q)^2 + \frac{\kappa^2}{2 \pi^2} \ln\left[ \frac{\abs{ (p-q)^2 }}{2 \ell_\mathrm{Pl}^2} \right] = 2 \sigma(p,q) + \frac{8 \ell_\mathrm{Pl}^2}{\pi} \ln\left[ \frac{\abs{ \sigma(p,q) }}{\ell_\mathrm{Pl}^2} \right] \eqend{,}
\end{equation}
where we chose $c = \mathe^{2 (1-\gamma)}/4$ and used that $\kappa^2 = 16 \pi \ell_\mathrm{Pl}^2$.

\subsection{Causal Structure}
\label{subsec:causal}

The causal relation between points is simpler to obtain, and based on the symplectic structure $\sigma$~\eqref{eq:sigma_def}. We first note that in the formal limit $f_{(\mu)}(x) \to v_\mu \delta^4(x-p)$ and $g_{(\mu)}(x) \to w_\mu \delta^4(x-q)$ we obtain
\begin{equation}
\sigma(f,g) \to - \frac{\kappa^2}{8 \pi} \eta^{\mu\nu} v_\mu w_\nu \sgn(p^0-q^0) \Theta\left[ -(p-q)^2 \right] \eqend{,}
\end{equation}
which vanishes if $p$ and $q$ are spacelike separated, and otherwise is constant with the sign depending on $\eta^{\mu\nu} v_\mu w_\nu$ and whether $p$ lies in the future of $q$ or the other way around. Summing again over an orthonormal Lorentzian frame $v_\mu = w_\mu = e_\mu^{(a)}$ with $a = 0,\ldots,3$ and $e_\mu^{(a)} = \delta_\mu^a$, we obtain
\begin{equation}
\label{eq:causal_sum_sigma}
- \frac{\kappa^2}{8 \pi} \sum_{a=1}^4 \eta_{aa} \eta^{\mu\nu} e_\mu^{(a)} e_\nu^{(a)} \sgn(p^0-q^0) \Theta\left[ -(p-q)^2 \right] = - \frac{\kappa^2}{2 \pi} \sgn(p^0-q^0) \Theta\left[ -(p-q)^2 \right] \eqend{,}
\end{equation}
whose sign now only depends on the causal relation between $p$ and $q$.

As for the distance, the correct functional to obtain the causal relation in the quantum theory is $\tau$ and not the state $\omega$, and so we define
\begin{definition}
Given positive test functions $\chi_p, \chi_q \in \mathbb{C}_0^\infty(\mathbb{R}^4)$ with mean $1$, their causal relation is encoded in the functional
\begin{equation}
\label{eq:calC_chip_chiq}
\mathcal{C}(\chi_p,\chi_q) \defeq - \frac{4 \pi \mathi}{\kappa^2} \eta_{ab} \ln \tau\left( W(f^{(a)}) W(g^{(b)}) W(- f^{(a)} - g^{(b)}) \right) = - \frac{2 \pi}{\kappa^2} \eta_{ab} \, \sigma(f^{(a)}, g^{(b)})
\end{equation}
with
\begin{equation}
f^{(a)}_\mu(x) = e_\mu^{(a)} \chi_p(x) \eqend{,} \quad g^{(a)}_\mu(x) = e_\mu^{(a)} \chi_q(x) \eqend{,}
\end{equation}
and where $\{ e_\mu^{(a)} \}_{a=1}^4$ is an orthonormal Lorentzian frame.

We say that $p$ and $q$ are spacelike separated if $\mathcal{C}(\chi_p,\chi_q) = 0$, that $p$ is in the future of $q$ if $\mathcal{C}(\chi_p,\chi_q) > 0$ and that $p$ is in the past of $q$ if $\mathcal{C}(\chi_p,\chi_q) < 0$.
\end{definition}
Since the support of $\chi_p$ might intersect the support of $\chi_q$ even if $p$ and $q$ are distinct, the causal relation between $p$ and $q$ is ``fuzzy''. Namely, we have
\begin{proposition}
In general, it holds that $\mathcal{C}(\chi_p,\chi_q) \in [-1,1]$. In the idealized limit $\chi_p(x) \to \delta^4(x-p)$ where the localization becomes exact we obtain $\mathcal{C}(\chi_p,\chi_q) \in \{-1,0,1\}$.
\end{proposition}
\begin{proof}
We compute
\begin{splitequation}
\mathcal{C}(\chi_p,\chi_q) &= \frac{1}{4} \iint \eta_{ab} \eta^{\mu\nu} f^{(a)}_{(\mu)}(x) g^{(b)}_{(\nu)}(x') \sgn(t-t') \Theta\left[ -(x-x')^2 \right] \total^4 x \total^4 x' \\
&= \iint \chi_p(x) \chi_q(x') \sgn(t-t') \Theta\left[ -(x-x')^2 \right] \total^4 x \total^4 x' \eqend{,}
\end{splitequation}
where we used that $\eta_{ab} \eta^{\mu\nu} e_\mu^{(a)} e_\nu^{(b)} = 4$, and from this
\begin{splitequation}
\abs{ \mathcal{C}(\chi_p,\chi_q) } &\leq \iint \abs{ \chi_p(x) \chi_q(x') \sgn(t-t') \Theta\left[ -(x-x')^2 \right] } \total^4 x \total^4 x' \\
&\leq \iint \chi_p(x) \chi_q(x') \total^4 x \total^4 x' = 1 \eqend{,}
\end{splitequation}
since $\chi_p(x)$ and $\chi_q(x)$ are positive functions with mean $1$. In the limit $\chi_p(x) \to \delta^4(x-p)$, it easily follows that
\begin{equation}
\mathcal{C}(\chi_p,\chi_q) \to \sgn(p^0-q^0) \Theta\left[ -(p-q)^2 \right] \in \{ -1, 0, 1 \} \eqend{.}
\end{equation}
\end{proof}
Values of $\mathcal{C}(\chi_p,\chi_q)$ between $0$ and $1$ can thus be interpreted as $p$ being partially in the future of $q$, namely some parts of the support of $\chi_p$ could lie spacelike to the support of $\chi_q$, or even in the causal past. On the other hand, if $\mathcal{C}(\chi_p,\chi_q) = 1$ we know that every $x \in \operatorname{supp} \chi_p$ lies in the causal future of every $x' \in \operatorname{supp} \chi_q$, and we may say that $p$ is strictly in the future of $q$ (with the given precision determined by $\chi_p$ and $\chi_q$).

In contrast to the DFR spacetime~\eqref{eq:dfr} with fixed noncommutativity, our proposal thus respects a weakened form of causality: a causal structure still exists, but can be fuzzy, and strict causality is only recovered in an idealized limit where one can localize points strictly. This is analogous to results in other approaches, for example $\kappa$-Minkowski spacetime~\cite{nevesfarinacougopinto2010,mercatisergola2018,mercatisergola2018b} (see also Refs.~\cite{francowallet2022,francohersentmariswallet2023} for discussions on causality in $\kappa$-Minkowski spacetime in the context of Lorentzian spectral triples). However, there are differences in the way this fuzzy causality is obtained, and consequently in its interpretation. Namely, the central object that is studied in these references on $\kappa$-Minkowski spacetime are the singular functions associated to the (deformed) Klein--Gordon equation, the retarded and advanced propagators and the Pauli--Jordan function, obtained either by solving the deformed equation or as expectation value of the commutator of two quantized scalar fields on $\kappa$-Minkowski spacetime. Classically, these functions have support inside the light cone and vanish outside of it, but the deformation smoothes out the jump at the light cone. Therefore, a fuzzy light cone can be defined as the region where this smooth transition happens. As noted in~\cite{mercatisergola2018b}, the size of this fuzzy cone generally depends on the quantum state of the quantized scalar field, but for semiclassical states is proportional to $\sqrt{ \ell_\mathrm{Pl} L }$, where $L$ is the distance to the origin. In contrast, our result~\eqref{eq:calC_chip_chiq} depends only on the symplectic structure that enters the quantization of the commutation relations~\eqref{eq:coord_commutator} and the localization of events through the smearing functions $\chi_p$, but not on the Dereziński--Meissner quantum state that is used to construct the Hilbert space. As explained in the introduction, our construction also does not make reference to particles or fields that propagate on the spacetime, such that the fuzzy causal structure that we obtain is an intrinsic property of the spacetime described by the noncommutative coordinates $X^{(\mu)}$. Moreover, there is a shift of perspective: instead of defining a fuzzy light cone that is experienced by fields propagating in spacetime, the causality function $\mathcal{C}$ encodes a fuzzy causality relation between events in spacetime.

\section{Discussion and outlook}
\label{sec:discussion}

In this work, we presented a quantum theory of a noncommutative spacetime that --- in contrast to some existing approaches --- preserves the Lorentzian distance and causal relations of Minkowski space in a fuzzy way, with Planck-scale corrections emerging naturally from the localization process intrinsic to the measurement of coordinates. We started with the construction of a Weyl algebra whose generators correspond to the noncommutative coordinates, based on the symplectic form derived from the Lorentz-invariant commutator that we obtained in previous work. We showed how to construct a state on this algebra, generalizing the Dereziński--Meissner construction of a quasi-free state for the massless two-dimensional scalar field. To ensure positivity, we employed a Krein space formulation, where the positive-definite scalar product that determines the state is used in the GNS construction of an auxiliary Hilbert space that is required for the quantization. This algebra encodes our noncommutative spacetime.

Adapting the Ehlers--Pirani--Schild philosophy to our setting, we defined a natural distance functional $\mathcal{D}(\chi_p, \chi_q)$ on the algebra, whose finiteness requires the introduction of suitably localized test functions $\chi_p$ that correspond to the finite resolution of the coordinate measurement of the point $p$, instead of the sharply localized points in the classical theory. The form and size of these test functions depend on the measurement apparatus, and their detailed characterization requires a modelling of this apparatus, which is beyond the scope of this work. After taking the classical limit, we may also perform the limit of strict localization, in which we recover the standard Minkowski distance (or Synge world function) from our distance functional. On the other hand, trying to strictly localize points in the quantum theory shows that there a exists an effective, dynamical minimal distance between them. Namely, since no measurement apparatus will be able to resolve distances better than the Planck length without collapsing into a black hole, we may choose for the test functions $\chi_p$ a normalized Gaussian with variance proportional to the squared Planck length. We then obtain a quantum correction to the Synge world function of the form
\begin{equation}
\sigma(p,q) \to \sigma(p,q) + \frac{4 \ell_\mathrm{Pl}^2}{\pi} \ln\left[ \frac{\abs{ \sigma(p,q) }}{\ell_\mathrm{Pl}^2} \right] + \mathcal{O}\left( \ell_\mathrm{Pl}^4 \right) \eqend{,}
\end{equation}
which is reminiscent of earlier proposals for a minimal length. However, the scale and detailed form of the correction depends on the chosen localization profile and the points being probed, rather than being imposed as a rigid universal cutoff. In this sense, the quantized coordinate system naturally exhibits fuzzy minimal length effects as quantum corrections to distance measurements, and the continuum geometry is recovered in the limit $\ell_\mathrm{Pl} \to 0$.

On the other hand, the causal structure is extracted directly from the symplectic form of the Weyl algebra as a causal functional $\mathcal{C}(\chi_p,\chi_q)$, evaluated on the same localized test functions $\chi_p$ as the distance functional. In the strict localization limit, $\mathcal{C}$ reduces to the usual division between spacelike ($\mathcal{C} = 0$) and timelike ($\mathcal{C} = \pm 1$) separation. In this limit, the noncommutativity of the coordinates also vanishes outside the light cone and becomes constant for timelike separations, such that microcausality is respected. In the quantum theory (for finite-width localizations) instead, the causal relation becomes intrinsically fuzzy: one obtains a continuous measure $\mathcal{C} \in [-1,1]$ of ``how much'' one measured point lies in the future of another, with values strictly between $-1$ and $1$ reflecting partial overlap of spacelike and timelike regions in the measurement. Crucially, however, this doesn't lead to acausal behavior: the fuzziness is induced by the measurement, while the algebra elements evolve causally, with their commutators vanishing for spacelike separations. The underlying fundamental causal structure is thus stable under quantization, but the measured causal relations are fuzzy.

From a broader perspective, our results provide a concrete, mathematically controlled realization of the idea that classical spacetime geometry can emerge from more fundamental quantum algebraic data. We have shown explicitly how a Lorentzian distance and causal order can be reconstructed from an abstract noncommutative algebra and a quantum state defined on this algebra, with the Planck-scale structure entering only as higher-order corrections to otherwise classical quantities. At the same time, our construction offers practical tools for investigating quantum gravitational effects on spacetime geometry: the Lorentzian noncommutative distance functional, the fuzzy causal relation, and the operator realization of noncommutative coordinates all admit systematic expansions in the deformation parameter $\kappa$, which is proportional to the Planck length $\ell_\mathrm{Pl}$. This makes the size and qualitative features of quantum corrections transparent, and provides a concrete model for comparison with phenomenological studies, such as~\cite{scardiglicasadio2014,buoninfanteetal2022,lambiasescardigli2022,scardiglilambiase2022,casadiofengkuntzscardigli2023}. Clearly, such a comparison will also need a model of the measurement apparatus.

Moreover, our results might lead the way for a resolution of the problem of spacetime topology discussed in~\cite{papadopoulosscardigli2018,papadopoulos2021,finstermuchpapadopoulos2021}. While in many cases, the topology that one equips a spacetime with is not a problem that needs an answer, and the underlying manifold topology can be used even when one studies the convergence of causal curves, some basic questions relating the topology and the spacetime structure cannot be answered in fully. These questions include in particular the inability of the standard manifold topology to encapsulate the causal structure: What is the relationship between the light cone and the manifold topology? How is the Lorentz group related to the group of homeomorphisms of the spacetime under the manifold topology? These questions have appeared as early as the first singularity theorem of Penrose~\cite{penrose1965}, and have been further studied by Zeeman~\cite{zeeman1967}, Göbel~\cite{goebel1976} and Hawking and McCarthy~\cite{hawkingkingmccarthy1976}. While the path topology that was introduced in the latter work encapsulates all possible structural levels of relativistic spacetimes (causal, conformal, differential, and so on), it does not satisfy the Limit Curve Theorem. However, this theorem is required to prove most (if not all) of the singularity theorems and their extensions, see for example~\cite{penrose1965,tipler1978,senovilla1998,senovillagarfinkle2015,papadopoulos2021,freivogelkontoukrommydas2022,fewsterkontou2022,grafetal2025}, and thus also the path topology cannot be the final answer to these questions. The topologization issue was brought again to the surface by Bykov and Minguzzi in~\cite{bykovminguzzi2025}, where they characterized global hyperbolicity outside of the frame of the manifold topology.

Since the question \emph{Which is the most appropriate topology for a spacetime?} has not been answered fully yet, and this problem remains open in relativity theory, we believe that for its solution it could be useful to explore new avenues, and in particular see which structures remain at the level of a quantum theory of spacetime. The aforementioned path topology is based in its construction on timelike paths, which do not exist anymore at the quantum level. What remains (at least in our proposal) is the causal relation $\mathcal{C}(\chi_p,\chi_q) \in [-1,1]$ between points $p$ and $q$ (with the measurement resolution described by the test function $\chi_p$), which is in general a fuzzy order relation where $p$ might be partially in the future of $q$. Only in the strict localization limit where $\chi_p(x) \to \delta^4(x-p)$ one recovers the strict causal ordering, where $\mathcal{C}(\chi_p,\chi_q) = 0$ if $p$ and $q$ are spacelike separated, $\mathcal{C}(\chi_p,\chi_q) = 1$ if $p$ is in the future of $q$ and $\mathcal{C}(\chi_p,\chi_q) = -1$ if $p$ is in the past of $q$, and thus an order relation between points. From the fuzzy order in the quantum theory, one could construct a fuzzy topology such as the $I$-fuzzy Alexandrov topology~\cite{bodenhofer2007,jinming2007}, which can be a starting point for a classification of appropriate spacetime topologies at the quantum level. As a first step, one then would like to recover the classical Alexandrov (interval) topology of the classical Minkowski spacetime. A different avenue would be a connection to Lorentzian metric spaces studied by Minguzzi, Suhr and collaborators~\cite{minguzzisuhr2024,bykovminguzzisuhr2025,minguzzi2025}, which defines a spacetime starting from the Lorentzian distance. In the quantum theory, this is extended by our distance functional $\mathcal{D}(\chi_p,\chi_q)$, which in the combined classical and strict localization limit reduces to the Lorentzian distance in Minkowski spacetime. We leave the exact construction of a fuzzy topology and its classical limit, as well as the relation to Lorentzian metric spaces, for future work.

\appendix

\section*{Acknowledgements}

\parindent=0pt
MBF thanks Leonardo Sangaletti for discussions. AM is grateful to Rainer Verch for stimulating discussions concerning the Ehlers--Pirani--Schild framework. We thank the anonymous reviewers for constructive criticism and useful suggestions.

\section*{Data Availability Statement}

Data sharing is not applicable to this article as no new data were created or analyzed in this study.

\clearpage

\bibliographystyle{elsarticle-num}
\bibliography{literature}
 
\end{document}